\documentclass[12pt,a4]{article}
\usepackage{amssymb}
\usepackage{graphicx}
\usepackage[all]{xy}
\usepackage{color,empheq}
\usepackage{latexsym}
\usepackage{amssymb,amsmath,amstext}
\usepackage{epsfig}
\usepackage{graphics}
\usepackage{amsthm}
\usepackage{verbatim}
\usepackage{mathrsfs}
\usepackage{a4wide}
\usepackage{hyperref}
\usepackage{pgfplots}

\newcommand{\Did}{\mathscr{D}}

\newcommand{\K}{\mathbb{Q}}
\newcommand{\Fgenh}{\mathbf F^\infty}
\newcommand{\Fgena}{\mathbf F}

\newcommand{\F}{\mathbf{F}}
\newcommand{\dwit}{\mathrm{d_{wit}}}

\DeclareMathOperator{\HS}{\mathsf{HS}}

\DeclareMathOperator{\wHS}{\mathsf{wHS}}
\DeclareMathOperator{\LM}{\mathsf{LM}}
\DeclareMathOperator{\MaxM}{\mathsf{MaxMinors}}

\DeclareMathOperator{\wdeg}{deg}

\DeclareMathOperator{\Icrit}{\mathbf{I}}

\DeclareMathOperator{\DEG}{\mathsf{DEG}}

\newcommand{\N}{\mathbb{N}}
\newcommand{\Sym}{\mathsf{Sym}}
\newcommand{\Z}{\mathbb{Z}}
\newcommand{\Q}{\mathbb{Q}}
\newcommand{\R}{\mathbb{R}}

\newcommand{\C}{\mathbb{C}}
\DeclareMathOperator{\jac}{\mathsf{jac}}

\DeclareMathOperator{\dreg}{d_{reg}}
\DeclareMathOperator{\dmax}{dmax}
\DeclareMathOperator{\GF}{\mathsf{GF}}

\newcommand{\frakg}{g^\infty}
\newcommand{\frakf}{f^\infty}

\newcommand{\frakq}{q^\infty}

\newcommand{\frakga}{g}
\newcommand{\frakfa}{f}

\newcommand{\frakqa}{q}

\newcommand{\frakghom}{g^h}
\newcommand{\frakfhom}{f^h}

\newcommand{\frakqhom}{q^h}

\newcommand{\Shom}{S^h}

\numberwithin{equation}{section}
\theoremstyle{plain}%
\newtheorem{theorem}{Theorem}
\numberwithin{theorem}{section}
\newtheorem{proposition}[theorem]{Proposition}

\newtheorem{lemma}[theorem]{Lemma}
\newtheorem{corollary}[theorem]{Corollary}
\newtheorem{definition}[theorem]{Definition}

\newtheorem{notation}[theorem]{Notation}

\date{}
\title{On the Complexity of Computing Critical Points with Gr\"obner Bases}
\author{Pierre-Jean Spaenlehauer\thanks{Project-team CARAMEL, INRIA/LORIA/CNRS, Nancy, France. Max Planck Institute for Mathematics, Bonn, Germany. PolSys Project-Team, INRIA/UPMC/LIP6/CNRS, Paris, France. pierre-jean.spaenlehauer@inria.fr}}

\begin{document}

\maketitle

\begin{abstract}
  Computing the critical points of a polynomial function
  $q\in\Q[X_1,\ldots,X_n]$ restricted to the vanishing locus
  $V\subset\R^n$ of polynomials $f_1,\ldots, f_p\in\Q[X_1,\ldots,
  X_n]$ is of first importance in several applications in optimization
  and in real algebraic geometry. These points are solutions of a
  highly structured system of multivariate polynomial equations
  involving maximal minors of a Jacobian matrix. We investigate the
  complexity of solving this problem by using Gr\"obner basis
  algorithms under genericity assumptions on the coefficients of the
  input polynomials. The main results refine known complexity bounds
  (which depend on the maximum
  $D=\max(\deg(f_1),\ldots,\deg(f_p),\deg(q))$) to bounds which depend
  on the list of degrees $(\deg(f_1),\ldots,\deg(f_p),\deg(q))$: we prove that
  the Gr\"obner basis computation can be performed in $\delta^{O(\log(A)/\log(G))}$
  arithmetic operations in $\Q$, where $\delta$ is the algebraic degree of the
  ideal vanishing on the critical points, and $A$ and $G$ are the arithmetic and geometric average of
  a multiset constructed from the sequence of degrees. As a
  by-product, we prove that solving such generic optimization problems
  with Gr\"obner bases requires at most $D^{O(n)}$ arithmetic
  operations in $\Q$, which meets the best known complexity bound for
  this problem. Finally, we illustrate these complexity results with
  experiments, giving evidence that these bounds are relevant for
  applications.
\end{abstract}

\section{Introduction}

\subsection{Problem statement and motivations}
Let $p<n$ be two positive integers, $q, f_1,\ldots,
f_p\in\Q[X_1,\dots,X_n]$ be polynomials with rational coefficients and
$\Icrit(q,\F)\subset \Q[X_1,\ldots, X_n]$ be the ideal generated by
the maximal minors of the Jacobian matrix $\jac(q,f_1,\ldots, f_p)$
and by the polynomials $\F=(f_1,\ldots, f_p)$:
$$\Icrit(q,\F)=\langle \MaxM(\jac(q,\F))\rangle+\langle\F\rangle.$$ 
Also, let $V\subset\C^n$ denote the variety associated to the
polynomial system $f_1,\ldots, f_p$. If $V$ is smooth, then the ideal
$\Icrit(q,\F)$ is the set of polynomials vanishing on the critical
points of the function $q$ restricted to $V$.  Computing these
critical points is of first importance in a wide range of applications
in optimization and in geometry, since the real local extrema of $q$ under
the constraints $f_1=\dots=f_p=0$ are reached at such points.

Our strategy is to compute these critical points by a classical
solving strategy with Gr\"obner bases algorithms (Buchberger
\cite{Buc65}/$F_4$ \cite{Fau99}/$F_5$ \cite{Fau02} and FGLM
\cite{FauGiaLazMor93}): we compute a parametrization of the critical
points by the roots of a univariate polynomial via a lexicographical
Gr\"obner basis of $\Icrit(q,\F)$ \cite{Rou99}. Once such a
representation is obtained, numerical algorithms can provide certified
approximations of the critical points from the roots of the univariate
polynomial.

It has been noticed that Gr\"obner bases algorithms perform well in
practice on these families of polynomial systems and critical point
computations are intensively used in the two software {\tt
  Hexagon}\footnote{written by Aur\'elien Greuet,
  http://www.lifl.fr/\textasciitilde greuet/indexFR.html} (for global
optimization) and {\tt RAGlib}\footnote{written by Mohab Safey El Din, http://www-polsys.lip6.fr/\textasciitilde safey/RAGLib/} (for computational real algebraic geometry). Both of
these software rely on the implementation of the $F_5$ algorithm
\cite{Fau02} in the {\tt FGb} library\footnote{written by Jean-Charles
  Faug\`ere, http://www-polsys.lip6.fr/\textasciitilde
  jcf/Software/FGb/}. It is not the aim of this paper to propose new
algorithms: the goal is to explain the efficiency of this approach by
providing complexity bounds that reflect the experimental behavior of
these Gr\"obner bases computations.
 
\subsection{Main results}

Let $\delta=\DEG(\Icrit(q,\F))$ denote the degree of the ideal
$\Icrit(q,\F)$ vanishing on the complex critical points. 
Classical
results on the algebraic degree of polynomial optimization show that
$\delta$ is in general exponential in $n$ and in $p$. An important
exception is the case of quadratic programming: in that case, $\delta$
is polynomial in the number of variables $n$ (but still exponential in
the codimension $p$) and this phenomenon is related to the existence
of algorithms running in time polynomial in $n$ in this special case
\cite{barvinok1993,GriPas05,FauSafSpa12}.

The value $\delta$ is an important indicator of the algebraic
complexity of the problem: under genericity assumptions on the input
polynomials, the ideal $\Icrit(q,\F)$ is radical (all the critical
points have multiplicity $1$) and hence $\delta$ equals the number of
complex critical points. Even though the number of real critical
points is usually much lower than $\delta$, in general their
coordinates lie in a field extension of $\Q$ of degree
$\delta$. Consequently, the complexity of Gr\"obner bases algorithms
is lower bounded by this value and we wish to express all upper
complexity bounds in terms of this lower bound.  In this setting, the
main result of this paper is

\begin{theorem}\label{thm:thm1intro}
Let $q,f_1,\ldots, f_p\in\Q[X_1,\ldots,X_n]$ be polynomials of degrees $d_0,d_1,\ldots, d_p\in\N$.
Let $A$ (resp. $G$) be the arithmetic (resp. geometric) average of the multiset
$$\{d_1,\ldots, d_p,\underbrace{\underset{0\leq i\leq p}\max\{d_i-1\},\dots,\underset{0\leq i\leq p}\max\{d_i-1\}}_{n-p}\}.$$
Then, under genericity assumptions on the coefficients of $q,f_1,\ldots,f_p$, a lexicographical Gr\"obner basis of $\Icrit(q,\F)$ can be computed within
$$\delta^{O\left(\log(A)/\log(G)\right)}\text{ arithmetic operations in $\Q$}.$$
\end{theorem}

In particular, by rephrasing this theorem and by bounding above $\log(A)/\log(G)$, we obtain the following complexity estimates in terms of the input size:

\begin{corollary}
Let $D=\max\{\deg(f_1),\ldots,\deg(f_p),\deg(q)\}$ denote the maximum of the degrees of the input polynomials.
Under genericity assumptions on the coefficients of the input system, the complexity of computing a lexicographical Gr\"obner basis of $\Icrit(q,\F)$ by classical Gr\"obner bases algorithms requires at most
$$D^{O\left(n\right)}\text{ arithmetic operations in $\Q$}.$$
\end{corollary}
The bound $D^{O\left(n\right)}$ meets the best known bound for the exact computation of such critical points \cite[Thm.~8]{SafTre06}, \cite[Sec.~10.3]{DBLP:journals/corr/DinS13}. 

\subsection{Roadmap of the proof}
In order to compute a lexicographical Gr\"obner basis of $\Icrit(q,\F)$, we use a
classical approach: first, we compute a \emph{graded reverse
  lexicographical} (\emph{grevlex} for short) Gr\"obner basis of $\Icrit(q,\F)$, then we use the FGLM algorithm to convert this
Gr\"obner basis into a \emph{lexicographical} (\emph{lex}) Gr\"obner basis.
The complexity of FGLM is well-known
\cite[Prop.~4.1]{FauGiaLazMor93} and bounded above by
$\delta^{O(1)}$. Consequently, we focus on the complexity
of the first step of the solving process, namely the computation of the grevlex Gr\"obner basis.

Let $f_1,\ldots,f_p,m_1,\ldots, m_{\binom{n}{p+1}}$ be the generators
of $\Icrit(q,\F)$, \emph{i.e.} the input polynomials and the maximal
minors of the Jacobian matrix. It is known that a grevlex Gr\"obner basis can be obtained by computing the row echelon form of a matrix (the \emph{Macaulay matrix}), which is parametrized by a degree $d\in\N$. Its row span is the $\Q$-vector space
$$T_d=\left\{\sum_{i=1}^p f_i \alpha_i +\sum_{j=1}^{\binom{n}{p+1}} m_j \beta_j \mid \alpha_i,\beta_j\in
\Q[X_1,\ldots, X_n], \deg(f_i \alpha_i)\leq d,\deg(m_j \beta_j)\leq
d\right\}.$$

The minimal value of $d$ such that a grevlex Gr\"obner basis of
$\Icrit(q,\F)$ is included in $T_d$ is called the \emph{witness
  degree} $\dwit$ of the system (see \emph{e.g.} \cite[Sec.~2.2]{BarFauSalSpa11}).  A crucial step in the
proof of the main complexity result is to estimate $\dwit$: under
genericity assumptions on the coefficients of $q,f_1,\ldots, f_p$
(Corollary \ref{coro:dreg}), we show the inequality
$$\dwit\leq (n-p-1) \max_{0\leq i\leq p}\{d_i-1\} -n-p+d_0 + 2\sum_{1\leq i\leq p} d_i.$$

To obtain this inequality, we first
consider the case where all polynomials $q,f_1,\ldots, f_p$ are
homogeneous with generic coefficients. In that case, the witness
degree equals the \emph{degree of regularity}, \emph{i.e.} the
smallest degree $d\in \N$ where the Hilbert function of the graded ring
$\Q[X_1,\ldots, X_n]/\Icrit(q,\F)$ becomes zero.

Next, we use techniques from commutative algebra to derive an explicit
formula for this Hilbert function.  This formula is obtained by
isolating the determinantal part of the ideal (Corollary
\ref{coro:HSImixed}) and by using the Eagon-Northcott complex
\cite{eagon1962ideals}\cite[Appendix A2H]{Eis01} to analyze this
determinantal component (Proposition \ref{prop:HSDmixed}). In fact,
under genericity assumptions on the coefficients of $(q,\F)$, the
Eagon-Northcott complex associated to the Jacobian matrix $\jac(q,\F)$
provides a \emph{graded free resolution} of the ideal generated by its
maximal minors, from which the Hilbert series can be extracted.

The inequality on the witness degree allows us to bound the dimension of the $\Q$-vector space
$T_{\dwit}$: its dimension is at most $\binom{(n-p-1) \max_{0\leq i\leq p}\{d_i-1\}-p+d_0 + 2\sum_{1\leq i\leq p} d_i}{n}$. A grevlex
Gr\"obner basis can be obtained by performing linear algebra in $S_{\dwit}$: this can be done within
$$O\left(\left(p+\binom{n}{p+1}\right)\binom{n+\dwit}{n}^\omega\right)$$
operations in $\Q$, where $\omega$ is a feasible exponent for the
matrix multiplication ($\omega\leq 2.373$ with Williams' algorithm
\cite{Vas11}). Rewriting this complexity bound in terms of the input degrees yields the claimed complexity bounds.

\subsection{Related works}
The results in this paper generalize the main results
in~\cite{FauSafSpa12} which were restricted to the special case
$q(X_1,\ldots,X_n)=X_1$ (\emph{i.e.} $d_0=1$) and
$d_1=\dots=d_p=D\in\N$. The complexity analysis was simpler in that
case: assuming that the degrees of the input polynomials are equal
simplifies the combinatorial and algebraic structure of the ideal
vanishing on the critical points. In this paper, we propose other
algebraic tools in order to take into account the combinatorial
structure induced by the mixed grading due to the different degrees of
general input polynomials. In particular, ideals generated by maximal
minors play an important role. When the grading is not uniform,
general formulas for the Hilbert series and for the Castelnuovo
regularity of such ideals are derived in \cite{budur2004hilbert}.

\smallskip

A classical problem in optimization is to compute a minimizer of a polynomial program of the form
\begin{equation}
  \left\{\begin{array}{l}
\text{Compute }\underset{(X_1,\ldots, X_n)\in\R}{\min} q(X_1,\ldots, X_n)\text{ under the constraints}\\
f_1(X_1,\ldots, X_n)=\dots=f_p(X_1,\ldots, X_n)=0.
\end{array}\right.
\label{eq:optineq}
\end{equation}
Inequalities can also be added to the set of constraints but this does
not change the algebraic degree of the problem if the minimizer lies
in the interior of the feasible set. Such a minimizer is a critical point of $q$
restricted to the variety associated to $f_1, \ldots, f_p$. The
explicit formula for the algebraic degree of this problem is given in
\cite[Theorem 2.2]{NieRan09} under genericity assumptions on the input
polynomials:
$$\delta=\displaystyle\left(\prod_{1\leq i\leq p} d_i\right)
\sum_{i_0+\dots+i_p=n-p} (d_0-1)^{i_0} \dots (d_p-1)^{i_p}.$$

Another area where such optimization problems appear frequently is computational real algebraic geometry.  For instance, the
\emph{critical point method} is a general
algorithmic framework for the study of topological properties of real
algebraic varieties. The cornerstone of these methods is the
computation of critical points of projections of varieties on linear
subspaces and have led to efficient algorithms with optimal or
near-optimal complexity for solving problems in real algebraic
geometry
\cite{GriVor88,Can88,BasPolRoy98,GriPas05,barvinok1993}. \emph{Polar
  varieties} describe the geometry of these critical loci and have
also given rise to large families of algorithms
\cite{trang1981varietes,BGHM1,BGHM2,BGHM3,BGHM4,SafSch03,BanGiuHeiPar05,BanGiuHeiSafSch10}. In
particular, \cite{BanGiuHeiMoh13} gives complexity bounds that are
polynomial in a geometrically defined quantity for solving these
polynomial optimization problems with geometric resolution techniques
\cite{GiuLecSal01}, leading to a complexity bound $(n D)^{O(n)}$ in the worst case.

The critical points of a function under polynomial constraints can
also be described as the set of solutions of a \emph{bi-homogeneous
  system} by using \emph{Lagrange multipliers}. This
connection between determinantal systems and bi-homogeneous systems
appears quite frequently: for instance in the Room-Kempf
desingularisation of determinantal varieties \cite{room1938geometry}
(see also \cite{bank2009bipolar}), or in the Kipnis-Shamir modeling of
the MinRank problem \cite{KipSha99,FauLevPer08,FauSafSpa10a}. This
modeling was used to obtain complexity bounds $D^{O(n)}$ by using the
geometric resolution algorithm \cite[Thm.~8]{SafTre06}. In the context of generalized Lagrange systems, similar complexity bounds are proved in \cite[Sec.~10.3]{DBLP:journals/corr/DinS13}.

Other applications of critical point computations with algebraic
methods appear in Statistics and in Biology
\cite{hosten2005solving,catanese2006maximum,pachter2005algebraic}. In
particular, computing \emph{Maximum Likelihood Estimates} is an
important routine in algebraic statistics which involves computing the
critical points of a monomial function restricted to an algebraic
variety. Another related setting is the study of the critical points
of the Euclidean distance function on an algebraic variety; algebraic
properties of these critical points are investigated in
\cite{DraHorOttStuTho13}.

\subsection{Organization of the paper}
The main notations used throughout the paper are introduced in Section \ref{sec:prelim}. Known results on the algebraic structure of the ideal $\Icrit(q,\F)$ are also recalled. The Eagon-Northcott complex is described in Section~\ref{sec:eagonnorthcott}, and formulas for the Hilbert series and for the degree of regularity are obtained in Section \ref{sec:HScritpointsmixed}; finally, the main complexity results are given in Section \ref{sec:critpointsmixedcompl} and are illustrated by experimental results in Section \ref{sec:expe}.

\subsection{Acknowledgements}
This is part of the author's Ph.D. thesis, written under the
supervision of Jean-Charles Faug\`ere and Mohab Safey El Din in the
PolSys project-team (INRIA/UPMC/LIP6, Paris, France). The author is
greatly indebted to Jean-Charles Faug\`ere, Mohab Safey El Din and
Bernd Sturmfels for their encouragements and for several helpful
discussions and suggestions. The author is grateful to Elisa Gorla for
pointing out important references and to an anonymous referee for useful
suggestions.

\section{Notations and preliminaries}
\label{sec:prelim}
\subsection{Notations}\label{sec:notations}
Throughout this paper, $p,n\in\N$ are two integers s.t. $p<n$, and $(d_0,\ldots, d_p)\in\N^{p+1}$ is a sequence of degrees such that
$d_0\geq 1$, $d_1,\ldots, d_p\geq 2$. Assuming that all the
constraints in the optimization problem are at least quadratic does not
lose any generality: linear constraints can be removed by substituting
one variable by a linear polynomial in the other equations. We let $X$ (resp. $U$) denote a
set of variables $\{X_1,\ldots, X_n\}$ (resp. $\{U_{0,1},\ldots,
U_{p,n}\}$) of cardinality $n$ (resp. $(p+1)n$).  We consider the
following grading of the polynomial ring $\Q[U,X]$:
$$\begin{cases}
  \deg(X_i)=1\text{ for all }i\in\{1,\ldots,n\};\\
  \deg(U_{i,j})=d_i-1 \text{ for all }i\in\{0,\ldots,p\},
  j\in\{1,\ldots,n\}.
\end{cases}$$ 

For $d\in\N$, we let $\Q[X]_d$ denote the $\Q$ vector space of homogeneous polynomials of degree $d$ in $\Q[X]$.  For any polynomial $\frakfa$ in $\Q[X]$, its
homogeneous part of highest degree is denoted by
$\frakf\in\Q[X]_{\deg(\frakfa)}$.  Also, for $i\in\{0,\ldots,p\},
j\in\{1,\ldots,n\}$, we define the polynomial $\frakg_{i
  n+j}\in\Q[U,X]$ as
$$\frakg_{i n+j}=
\begin{cases}
  U_{i,j}-\frac{\partial \frakf_i}{\partial X_j}\text{ if $i\in\{1,\ldots,p\}$},\\
  U_{i,j}-\frac{\partial \frakq}{\partial X_j}\text{ if $i=0$}.
\end{cases}
$$
For $i\in\{1,\ldots,p\}$, we define $\frakg_{(p+1)n+i}\in\Q[U,X]$ to
be equal to $\frakf_i\in\Q[X]$.
The determinantal ideal generated by all the maximal minors of the matrix
$$U=
\begin{bmatrix}
  U_{0,1}&\dots&U_{0,n}\\
  \vdots&\vdots&\vdots\\
  U_{p,1}&\dots&U_{p,n}
\end{bmatrix}$$ is denoted by $\Did$.  For $\F\in \Q[X]^{p}$ and an
objective function $q\in\Q[X]$ we consider the ideal $\Icrit(q,\F)$
generated by $\langle f_1,\ldots, f_p\rangle$ and by the maximal
minors of the Jacobian matrix $\jac(q,\F)$
$$
\jac(q,\F)=\begin{bmatrix} \frac{\partial q}{\partial
    X_1}&\dots&\frac{\partial q}{\partial
    X_n}\\
  \frac{\partial f_1}{\partial X_1}&\dots&\frac{\partial f_1}{\partial X_n}\\
  \vdots&\vdots&\vdots\\
  \frac{\partial f_{p}}{\partial X_1}&\dots&\frac{\partial
    f_{p}}{\partial
    X_n}\\
\end{bmatrix},$$
\emph{i.e.} $\Icrit(q,\F)=\langle f_1,\ldots, f_p\rangle +
\langle\MaxM(\jac(q,\F)\rangle$. Therefore,
$$\Icrit(\frakqa,\Fgena)=\left(\Did+\langle\frakga_1,\ldots,\frakga_{(p+1)n+p}\rangle\right)\cap\Q[X].$$

For a graded ideal $I$ of an $\N$-graded
$\Q$-algebra $R$, we call \emph{dimension} of $I$ the Krull
dimension of the quotient ring $R/I$, and we let $\wHS_{R/I}\in\Z[[t]]$ denote the
\emph{weighted Hilbert series} of $I$, defined by
$$\wHS_{R/I}(t)=\sum_{d\in\N} \dim_\Q\left(R_d/I_d\right)t^d,$$

where $R_d$ (resp. $I_d$) denotes the $\Q$-vector space of homogeneous
elements of degree $d$ in $R$ (resp. $I$). We shall use the notation $\HS_{R/I}(t)$ when $R$ is the polynomial algebra $\Q[X]$ (or $\Q[X]$) with the canonical grading $\deg(X_i)=1$ for all $i$.

If $I$ is an homogeneous ideal of dimension $0$ of a polynomial ring $\Q[X]$, its \emph{degree of regularity} $\dreg(I)$ is the smallest integer $d$ such that $\dim_\Q(\Q[X]_d/I_d)=0$. Equivalently, it equals $1$ plus the degree of the Hilbert series (which is a polynomial in the case of $0$-dimensional ideals).

\subsection{Genericity}
Throughout this paper, $q,f_1,\ldots, f_p\in \Q[X]$ are polynomials of respective degrees at most $d_0,\ldots, d_p$. 
As in \cite{NieRan09}, we say that a property holds for a \emph{generic system} $(q,f_1,\ldots, f_p)$ (resp. $(q^\infty,f_1^\infty,\ldots, f_p^\infty)\in\Q[X]_{d_0}\times\dots\times \Q[X]_{d_p}$) if this property holds for all $(q, f_1,\ldots, f_p)$ (resp. $(q^\infty,f_1^\infty,\ldots, f_p^\infty)$) in a dense Zariski open subset of the space of all polynomials of degrees at most $d_0,\ldots, d_p$ (resp. of the space of all homogeneous polynomials of degrees $d_0,\ldots, d_p$). Note that the variety associated to a generic system $f_1,\ldots, f_p$ (resp. $f_1^\infty,\ldots, f_p^\infty$) is a reduced smooth complete intersection.

\subsection{Complexity model and notations}
All complexity estimates count the number of operations $\{+,-,\times,\div\}$ in $\Q$. It is not the goal of this paper to estimate the bit complexity induced by the growth of the coefficients due to the arithmetic operations in $\Q$. All complexity counts are parametrized by $n$ and by the sequence of degrees $(d_0,\ldots, d_p)$ of length $p+1$. More precisely,  complexities in this model are partial functions $\N\times \mathbf c_0(\N)\rightarrow \N$, where $\mathbf c_0(\N)$ is the set of sequences $(d_i)_{i\in\N}$ with finite support. For two non-negative functions $C_1,C_2:\N\times \mathbf c_0(\N)\rightarrow \R_+$, we write $C_1=O(C_2)$ if there exists a constant $A\in \N$ such that $C_1(n,d_0,\ldots, d_p)\leq A\cdot C_2(n,d_0,\ldots, d_p)$. For three functions $C_1,C_2,C_3:\N\times \mathbf c_0(\N)\rightarrow \R_+$ taking values greater than $1$, the notation $C_1=C_2^{O(C_3)}$ means that
$$\frac{\log(C_1)}{\log(C_2)}=O(C_3).$$

\subsection{Algebraic structure of $\Icrit(\frakq,\Fgenh)$}

We recall in this section results from a previous paper
\cite{FauSafSpa12}, where we investigated the special case where all
the constraints shared the same degree: $d_1=\dots=d_p=D$ and $q=X_1$ is
the projection on the first coordinate. Some of their properties also
hold in the general case. We state them and recall their proofs in this section.
The next lemma shows that the set of complex critical points of a generic polynomial optimization problem is finite. It is stated for the homogeneous system $(q^\infty,\F^\infty)$ but the same statement for an inhomogeneous generic system $(q,\F)$ can be proved similarly (see \cite[Prop.~2.1]{NieRan09}).

\begin{lemma}\label{lem:dim0}
  If $(\frakq,\Fgenh)$ is a generic system, then the ideal $\Icrit(\frakq,\Fgenh)$ has dimension $0$.
\end{lemma}
\begin{proof}
This proof is similar to that of \cite[Lemma~2.1]{DraHorOttStuTho13}.
Let $V\subset \C^n$ be the affine variety associated to $\Fgenh$. Since $\Fgenh$ is generic, $V$ is smooth at any nonzero point. We consider the correspondence variety
$$\mathcal E_{V,d_0}=\C[x_1,\ldots,x_n]_{d_0}\times V\subset \C^{\binom{n+d_0-1}{n}+n}.$$
The set $W=\{(q^\infty,\mathbf x)\in \mathcal E_{V,d_0}\mid {\rm
  Rank}(\jac(q^\infty(\mathbf x),\F^\infty(\mathbf x))\leq p\}$ is a proper
subvariety of $\mathcal E_{V,d_0}$ and $\Icrit(\frakq,\Fgenh)$ is the generic fiber of its projection on $\C[x_1,\ldots,x_n]_{d_0}$. Since $V$ is smooth at any nonzero point, for all $\mathbf x\in V\setminus\{0\}$,
the fiber $\{q^\infty\in \C[x_1,\ldots,x_n]_{d_0}\mid (q^\infty,\mathbf x)\in W\}$
is a linear subspace of codimension $n-p$ in
$\C[x_1,\ldots,x_n]_{d_0}$. Since $V$ has dimension $n-p$ and the fibers of the projection on $V$ have codimension $n-p$,  $W$ must have dimension $\dim_\C(\C[x_1,\ldots,x_n]_{d_0})$ and hence the generic fiber of its projection on $\C[x_1,\ldots,x_n]_{d_0}$
is finite. 
\end{proof}

Since the ideal $\Icrit(\frakq,\Fgenh)$ is homogeneous and has dimension $0$, its variety is the unique point $\{0\}\in\C^n$. However, the ideal $\Icrit(\frakq,\Fgenh)$ is not radical and an important indicator of the complexity of the Gr\"obner basis computation is its Hilbert series. The two next statements describe the relationship between the Hilbert series of $\Q[X]/\Icrit(\frakq,\Fgenh)$ and that of $\Q[U]/\Did$.

\begin{lemma}\label{lem:nondivgl}
  For $\ell\in\{1,\ldots,p+n(p+1)\}$, the polynomial $\frakg_\ell$ does not divide $0$ in the quotient ring $\Q[U,X]/\left(\Did+\langle  \frakg_1,\ldots, \frakg_{\ell-1}\rangle\right)$.
\end{lemma}
\begin{proof}
   The quotient ring $\Q[U,X]/\Did$ is Cohen-Macaulay of dimension $p+n(p+1)$ \cite[Prop.~1.1]{BruVet88}, \cite[Coro.~2.8]{BruVet88} and $\Did+\left\langle  \frakg_1,\ldots, \frakg_{p+n(p+1)}\right\rangle$ has dimension $0$ in $\Q[U,X]$ by Lemma \ref{lem:dim0}. Consequently, $\langle  \frakg_1,\ldots, \frakg_{p+n(p+1)}\rangle$ is a regular sequence in $\Q[U,X]/\Did$ by Macaulay's unmixedness Theorem \cite[Corollary 18.14]{Eis95}.
\end{proof}

Homogeneous regular sequences provide relations between the Hilbert series of the associated ideals, leading to the following corollary:
\begin{corollary}\label{coro:HSImixed}
The Hilbert series of $\Q[X]/\Icrit(\frakq,\Fgenh)$ is related to that of $\Q[U]/\Did$ by
$$\HS_{\Q[X]/\Icrit(\frakq,\Fgenh)}(t)=\wHS_{\Q[U]/\Did}(t)\cdot\frac{(1-t^{d_{0}-1})^n\prod_{1\leq i\leq p} (1-t^{d_i})(1-t^{d_i-1})^n}{(1-t)^n}.$$
\end{corollary}

\begin{proof}
First, we use the fact that $\Q[U,X]/\Did \cong \Q[U]/\Did \otimes_\Q \Q[X]$ which yields
$$\wHS_{\Q[U,X]/\Did}(t)=\wHS_{\Q[U]/\Did}(t)\cdot\HS_{\Q[X]}(t)=\dfrac{\wHS_{\Q[U]/\Did}(t)}{(1-t)^n}.$$

We recall that, with the notations of Lemma \ref{lem:nondivgl},
$$\Icrit(\frakq,\Fgenh)=\left(\Did+\left\langle \frakg_1,\ldots, \frakg_{p+n(p+1)}\right\rangle\right)\cap\Q[X].$$

According to Lemma \ref{lem:nondivgl}, for each $2\leq\ell\leq p+n(p+1)$, $\frakg_\ell$ does not divide $0$ in $\Q[X]/\langle\frakg_1,\ldots,\frakg_{\ell-1}\rangle$. Adding to an ideal a polynomial of degree $d$ that does not divide zero in the quotient ring multiplies its Hilbert series by $(1-t^d)$.
To conclude the proof, we notice that 
$$\Q[U,X]/\left(\Did+\langle \frakg_1,\ldots, \frakg_{p+n(p+1)}\rangle\right)\cong \Q[X]/\Icrit(\frakq,\Fgenh)$$
via the map
$$
\left\{\begin{array}{rcl}
  X_i&\mapsto& X_i\\
  U_{i,j}&\mapsto&\dfrac{\partial f_i}{\partial X_j}
\end{array}\right.
$$
Consequently
$$\begin{array}{rcl}\HS_{\Q[X]/\Icrit(\frakq,\Fgenh)}(t)&=&\HS_{\Q[U,X]/\left(\Did+\langle \frakg_1,\ldots, \frakg_{p+n(p+1)}\rangle\right)}(t)\\
  &=&\displaystyle\wHS_{\Q[U,X]/\Did}(t) (1-t^{d_{0}-1})^n\prod_{1\leq i\leq p} (1-t^{d_i})(1-t^{d_i-1})^n\\
  &=&\displaystyle\wHS_{\Q[U]/\Did}(t)\cdot\dfrac{(1-t^{d_{0}-1})^n \prod_{1\leq i\leq p} (1-t^{d_i})(1-t^{d_i-1})^n}{(1-t)^n}.\end{array}$$
\end{proof}

\section{The Eagon-Northcott complex and the Hilbert series of weighted determinantal ideals}
\label{sec:mixedCritPoints}
The goal of this section is to obtain an
explicit formula for the degree of regularity of
$\Icrit(\frakq,\Fgenh)$, namely the smallest positive integer $\ell$
such that the coefficient of $t^\ell$ in the series expansion of
$\HS_{\Q[X]/\Icrit(\frakq,\Fgenh)}(t)$ is zero. This value bounds the witness degree of any system of generators of $\Icrit(q,\F)(t)$ and hence is an
indicator of the complexity of the Gr\"obner basis computation.

The bound on $\dwit$ shall be obtained via an
explicit formula for the Hilbert series of
$\Icrit(q^\infty,\F^\infty)$. 
The main principle is to isolate first the determinantal component
(generated by the maximal minors of the Jacobian matrix) and to analyze it
separately. We obtain in that case a determinantal ideal $\Did$, with a
grading induced by the degrees of the input polynomials.
A free resolution of this ideal
$\Did$ is given by the so-called \emph{Eagon-Northcott complex}
\cite{eagon1962ideals}. From this free resolution, we shall read off an
explicit formula for the Hilbert series and for the degree of
regularity of $\Icrit(q^\infty,\F^\infty)$.

\subsection[Eagon-Northcott complex]{Preliminaries on the Eagon-Northcott complex}
\label{sec:eagonnorthcott}

Free resolutions are classical tools in commutative algebra to
describe the relations existing between a given set of
polynomials. The main principles of these techniques go back to
Hilbert and his \emph{Syzygy Theorem} (see \emph{e.g.} \cite[Corollary
19.7]{Eis95} for a statement in the modern formalism).

An explicit description of a minimal free resolution of the ideal
generated by the maximal minors of a generic matrix is given by the
Eagon-Northcott complex \cite{eagon1962ideals}. We refer to
\cite[Appendix A2H]{Eis01} for a complete presentation.  After
describing the general construction of the complex, we shall detail an
example of Hilbert series computations.

Let $R=Q[U]$ be the polynomial ring. Following the notations in \cite[Appendix A2H]{Eis01}, we write
$F=R^f$ and $G=R^g$, where $f$ and $g$ are two integers such that $g<f$. For a $g\times f$ matrix whose
entries are in $R$, we let $\alpha:F\rightarrow G$ denote the
corresponding morphism of modules. Let $\otimes^i G$ be the $R$-module of tensors of order $i$:
$$\otimes^i G={\rm Span}_R\left\{g_1\otimes\dots\otimes g_i \mid g_1,\ldots, g_i\in G\right\}$$
and let $M_i$ (resp. $N_i$) be the submodule of $\otimes^i G$
generated by the elements $\{g_1\otimes\dots\otimes g_i -
g_{\sigma(1)}\otimes\dots\otimes g_{\sigma(i)} \mid g_1,\ldots, g_i\in G,
\sigma \text{ a permutation of }\{1,\ldots,i\}\}$
(resp. $\{g_1\otimes\dots\otimes g_i -
(-1)^{\mathsf{sgn}(\sigma)}g_{\sigma(1)}\otimes\dots\otimes
g_{\sigma(i)} \mid g_1,\ldots, g_i\in G, \sigma \text{ a permutation of
}\{1,\ldots,i\}\}$. Then we let $\Sym_i G=\left(\otimes^iG\right)/M_i$
(resp. $\bigwedge^i G=\left(\otimes^i G\right)/N_i$) be the $R$-module
of elements of order $i$ in the \emph{symmetric algebra} (resp. in the
\emph{exterior algebra}).  The Eagon-Northcott complex is then defined
by:
$$\begin{array}{rl}
  \mathsf{EN}(\alpha):&0\rightarrow (\mathsf{Sym}_{f-g} G)^*\otimes \bigwedge^f F\xrightarrow{\sigma_{f-g-1}}(\mathsf{Sym}_{f-g-1}G)^*\otimes \bigwedge^{f-1} F\xrightarrow{\sigma_{f-g}}\\
  &\dots\rightarrow (\mathsf{Sym}_2 G)^*\otimes \bigwedge^{g+2} F\xrightarrow{\sigma_3}G^*\otimes \bigwedge^{g+1} F\xrightarrow{\sigma_2}\bigwedge^g F\xrightarrow{\wedge^g\alpha}\bigwedge^g G.
\end{array}$$

First, notice that as a $R$-module, $\mathsf{Sym}_i G$ (and hence
also its dual $(\mathsf{Sym}_{i} G)^*$) is a free module isomorphic
to $R^{\binom{i+g-1}{i}}$. Similarly, $\bigwedge^i F$ is isomorphic to
$R^{\binom{f}{i}}$ as a $R$-module.
For a detailed description of the maps $\sigma_i$, we refer 
to \cite[Appendix A2H]{Eis01} and \cite[Appendix A2.6]{Eis95}. In the context of this paper, $f=n$, $g=p+1$,
$R=\Q[U]$ with the grading given by
$\wdeg(U_{i,j})=d_i-1$, and the map $\alpha$ corresponds to the
matrix
$$\mathscr U=\begin{pmatrix}
  U_{0,1}&\dots&U_{0,n}\\
  \vdots&\vdots&\vdots\\
  U_{p,1}&\dots&U_{p,n}
\end{pmatrix}.$$

The next step is to take into account the grading $\wdeg(U_{i,j})=d_i-1$. We
use here the classical notation $R(-s)$ to denote the ring $R$ where
the grading has been shifted by $s$, \emph{i.e.} $R(-s)_d=R_{d-s}$ . For instance, if $\eta: R\rightarrow R$ is a morphism of
degree $s$, the induced morphism
$$R(-s)\xrightarrow{\eta} R(0)$$
maps elements of degree $d$ in $R(-s)$ to elements of degree $d$ in $R$.
  Using the notation $s=\sum_{0\leq i\leq p}(d_i-1)$, and taking into account the
grading of $\Q[U]$ and the description of the maps in \cite[Appendix A2H]{Eis01}, the complex can be rewritten as
$$\begin{array}{@{}r@{~}l@{}}
  \mathsf{EN}(\alpha):&0\rightarrow \underset{\substack{i_0+\ldots+i_p\\=\\n-p-1}}{\bigoplus} R\left(-s-\underset{0\leq j\leq p}{\sum}i_j (d_j-1)\right)\xrightarrow{\sigma_{f-g-1}}\underset{\substack{i_0+\ldots+i_p\\=\\n-p-2}}{\bigoplus} R\left(-s-\underset{0\leq j\leq p}{\sum}i_j (d_j-1)\right)^{\binom{n}{n-1}}\\
  &\xrightarrow{\sigma_{f-g}} \underset{\substack{i_0+\ldots+i_p\\=~2}}{\bigoplus} R\left(-s-\underset{0\leq j\leq p}{\sum}i_j (d_j-1)\right)^{\binom{n}{p+3}} \xrightarrow{\sigma_3} \underset{1\leq i\leq p}{\bigoplus} R\left(-s-(d_i-1)\right)^{\binom{n}{p+2}}\xrightarrow{\sigma_2}\\&\rightarrow R(-s)^{\binom{n}{p+1}}\xrightarrow{\wedge^g\alpha}R(0).
\end{array}$$

The last map $\wedge^g\alpha$ sends each generator of $R(-s)^{\binom{n}{p+1}}$ on a maximal minor of $\alpha$ (note that the number of such maximal minors is precisely $\binom{n}{p+1}$), hence the image of $\wedge^g\alpha$ is indeed the ideal generated by the maximal minors.

\medskip

{\bf Example.}
Set $n=4$ and $p=1$. The ideal $\Did$ is generated by the maximal minors of the following $2\times 4$ matrix:
$$\begin{pmatrix}
  U_{0,1}&U_{0,2}&U_{0,3}&U_{0,4}\\
  U_{1,1}&U_{1,2}&U_{1,3}&U_{1,4}
\end{pmatrix}$$

In this case, the Eagon-Northcott complex is
$$\begin{array}{rl}
  \mathsf{EN}:&0\rightarrow R^3\xrightarrow{\sigma_{3}} R^8 \xrightarrow{\sigma_2} R^6\xrightarrow{\sigma_1} R
\end{array},$$

where the letter $R$ stands for the ring $\Q[U]$, and the morphisms $\sigma_i$ are given by the following matrices:
\footnotesize
$$\begin{array}{@{}c@{}}
\begin{array}{ccccc}\wedge\alpha=&(
  U_{0,2} U_{1,4} - U_{1,2} U_{0,4}, & U_{0,3} U_{1,2}-U_{0,2} U_{1,3}, & U_{0,1} U_{1,3} - U_{0,3} U_{1,1}, \\ &U_{0,4} U_{1,1} - U_{1,4} U_{0,1}, &U_{0,4} U_{1,3} - U_{1,4} U_{0,3}, & U_{0,2} U_{1,1} - U_{1,2} U_{0,1})
\end{array}\\
\begin{array}{rcl}
  \sigma_2&=&\begin{pmatrix}
    -U_{0,3}&-U_{1,3}&0&0&U_{0,1}&U_{1,1}&0&0\\
    -U_{0,4}&-U_{1,4}&0&0&0&0&U_{0,1}&U_{1,1}\\
    0&0&-U_{0,4}&-U_{1,4}&0&0&U_{0,2}&U_{1,2}\\
    0&0&-U_{0,3}&-U_{1,3}&U_{0,2}&U_{1,2}&0&0\\
    -U_{0,2}&-U_{1,2}&U_{0,1}&U_{1,1}&0&0&0&0\\
    0&0&0&0&-U_{0,4}&-U_{1,4}&U_{0,3}&U_{1,3}
  \end{pmatrix}\\
  \sigma_3&=&\begin{pmatrix}
    U_{0,1}&U_{1,1}&0\\
    0&U_{0,1}&U_{1,1}\\
    U_{0,2}&U_{1,2}&0\\
    0&U_{0,2}&U_{1,2}\\
    U_{0,3}&U_{1,3}&0\\
    0&U_{0,3}&U_{1,3}\\
    U_{0,4}&U_{1,4}&0\\
    0&U_{0,4}&U_{1,4}
  \end{pmatrix}
\end{array}
\end{array}$$

\normalsize Direct computations show that this is a complex (\emph{i.e.}
for all $i$, $\sigma_{i-1}\circ\sigma_i=0$) and $\mathsf{Im}(\sigma_1)=\Did$.
 Taking the grading into account, the Eagon-Northcott complex is rewritten as
$$\begin{array}{rl}
  \mathsf{EN}:&0\rightarrow R(-3d_0-d_1-4)\oplus R(-2d_0-2d_1-4)\oplus R(-d_0-3d_1-4) \xrightarrow{\sigma_{3}}\\& R(-2d_0-d_1-3)^4\oplus R(-d_0-2d_1-3)^4\xrightarrow{\sigma_2} R(-d_0-d_1-2)^6\xrightarrow{\wedge\alpha} R(0)\rightarrow R/\Did\rightarrow 0.
\end{array}$$

\subsection[Hilbert series, degree of regularity]{Hilbert series and degree of regularity}
\label{sec:HScritpointsmixed}
The next goal is to derive an explicit formula for the Hilbert series
of $\Icrit(q^\infty,\F^\infty)$. We use the fact that the Hilbert
series can be computed once a free resolution is known, since the
Hilbert series of $I$ is equal to the alternate sum of the Hilbert
series of the free modules occurring in the resolution. This is a
classical strategy for obtaining an explicit formula for the Hilbert
series (see \emph{e.g.} \cite[Theorem 1.11]{Eis01} for more details).
The next proposition is a special case of \cite[Prop.~2.4]{budur2004hilbert}, which gives an
explicit formula for the Hilbert series for more general gradings.

\begin{proposition}\label{prop:HSDmixed}
The weighted Hilbert series of the ideal $\Did\subset \Q[U]$ generated by the maximal minors of the matrix $\mathscr U$
with $\wdeg(U_{i,j})=d_i-1$ is the power series expansion of the rational function
$$\wHS_{\Q[U]/\Did}(t)=\frac{1-\left[\underset{0\leq k\leq n-p-1}{\sum}\left[(-1)^k\underset{i_0+\ldots+i_p=k}{\sum}\binom{n}{p+k+1}\cdot t^{\underset{0\leq j\leq p}{\sum}(i_j+1)(d_j-1)}\right]\right]}{\prod_{0\leq i\leq p} (1-t^{d_i-1})^n}.$$
\end{proposition}

\begin{proof}
  According to \cite[Theorem 1.11]{Eis01}, the Hilbert series of a graded ideal can be
  computed from a minimal free resolution: it equals the
  alternate sum of the Hilbert series of the free modules occurring in
  the resolution.  For $\ell,j\in \mathbb N$, the Hilbert series of
  $R(-\ell)^j$ equals
$$\wHS_{R(-\ell)^j}(t)=\frac{j t^\ell}{\prod_{0\leq i\leq p} (1-t^{d_i-1})^n}.$$

Moreover, the Hilbert series of a direct sum of modules is equal to the sum of their Hilbert series. Therefore, by considering the alternate sum of the Hilbert series of the free modules in the Eagon-Northcott complex (which is a free resolution of $\Did$), direct computations yield the formula for the weighted Hilbert series of $\Did$.
\end{proof}

The degree of regularity can be extracted from the Hilbert series, yielding the following formula:

\begin{corollary}\label{coro:dreg}
For generic homogeneous polynomials $(q^\infty,f_1^\infty,\ldots, f_p^\infty)$, the degree of regularity of $\Icrit(\frakq,\Fgenh)$ is 
$$\displaystyle\dreg(\Icrit(\frakq,\Fgenh)) =  (n-p-1) \underset{0\leq i\leq p}{\max}\{d_i-1\} -n-p+d_0 + 2\sum_{1\leq i\leq p} d_i.$$
\end{corollary}

\begin{proof}
  Since the ideal $\Icrit(\frakq,\Fgenh)$ is $0$-dimensional (Lemma
  \ref{lem:dim0}), $\HS_{\Q[X]/\Icrit(\frakq,\Fgenh)}(t)$ is a
  polynomial and
  $\dreg=\deg(\HS_{\Q[X]/\Icrit(\frakq,\Fgenh)}(t))+1$.  Let
  $j_0$ be the index of one of the maximal degrees:
  $d_{j_0}=\max_{0\leq j\leq p}\{d_j\}$.  In the sums in the numerator
  of the formula given in Proposition \ref{prop:HSDmixed}, the maximal
  degree is reached when $k=n-p-1$, $i_{j_0}=k$, and $i_j=0$ for
  $j\neq j_0$. Therefore the degree of the numerator of
  $\wHS_{\Q[U]/\Did}(t)$ equals
$$\begin{array}{ll}&\deg\left(1-\left[\underset{0\leq k\leq n-p-1}{\sum}\left[(-1)^k\underset{i_0+\ldots+i_p=k}{\sum}\binom{n}{p+k+1}t^{\underset{0\leq j\leq p}{\sum}(i_j+1)(d_j-1)}\right]\right]\right)\\ =&(n-p-1) (\max\{d_j\}-1) +\sum_{0\leq j\leq p} (d_j-1).\end{array}$$

On the other hand, we have
$$\begin{cases}
  \displaystyle\deg\left(\prod_{1\leq i\leq p} (1-t^{d_i})\right)=\sum_{1\leq i\leq p}d_i;\\
  \displaystyle\deg\left((1-t)^n\right)=n.
\end{cases}$$

Therefore, using the formula in Corollary \ref{coro:HSImixed}, we obtain
$$\begin{array}{rcl}
  \deg(\HS_{\Q[X]/\Icrit(\frakq,\Fgenh)})&=&\displaystyle (n-p-1) (\max\{d_i\}-1) +\sum_{0\leq i\leq p} (d_i-1) +
  \displaystyle \sum_{1\leq i\leq p}d_i -n\\
  &=&\displaystyle (n-p-1) \max\{d_i-1\} -n-p+d_0-1 + 2\sum_{1\leq i\leq p} d_i,\end{array}$$

and hence $\displaystyle\dreg =  (n-p-1) \max\{d_i-1\} -n-p+d_0 + 2\sum_{1\leq i\leq p} d_i.$
\end{proof}

\subsection{Grothendieck polynomials}
In this section, we discuss briefly another approach to compute the Hilbert series of weighted determinantal via Grothendieck polynomials introduced in \cite{LasSch82}.

The numerator of the rational function in Proposition
\ref{prop:HSDmixed} -- also called $K$-polyno\-mial - is equal to the
evaluation of Grothendieck polynomials at powers of $t$, see
\cite[Theorem A]{KnuMil05}. On the other hand, the
\emph{(multi-)degree} of the ideal may be expressed in the evaluation
of Schubert polynomials \cite[Theorem 15.40]{MilStu05}. In the sequel,
$S_{n+1}$ denotes the group of permutations of the set $\{1,\ldots,
n+1\}$ and for all $i$, $\sigma_i$ is the transposition
$i\leftrightarrow i+1$.

\begin{definition}
The divided difference operators $\partial_i$ are defined by
$$\forall H\in\Z[t_1,\ldots, t_{n+1}], \partial_i H=\dfrac{H(t_1,\ldots, t_{n+1})-H(t_1,\ldots,t_{i-1},t_{i+1},t_i,t_{i+2},\ldots t_{n+1})}{t_i-t_{i+1}}.$$

Let $w_0\in S_{n+1}$ be the permutation $w(i)=n+2-i$.
For $w\in S_{n+1}$, the \emph{Grothendieck polynomial} $\mathcal G_w\in\Z[t_1,\ldots, t_{n+1}]$ is defined by
$$
\begin{array}{rcl}
  \mathcal G_{w_0}(t_1,\ldots, t_{p+1})&=&\displaystyle\prod_{i=1}^{n+1}\left(1-t_i\right)^{n+1-i},\\
  \mathcal G_{w\cdot\sigma_i}(t_1,\ldots, t_{p+1})&=&-\partial_i \left(t_{i+1}\mathcal G_w(t_1,\ldots, t_{p+1})\right) \text{ when }{\rm length}(w\cdot\sigma_i)<{\rm length}(w).
\end{array}
$$
\end{definition}

Let $S_{n+1}$ be the group of permutations on the set $\{1,\ldots, n+1\}$. To the determinantal ideal $\Did$ is associated the following permutation $w\in S_{n+1}$ (see \cite[Chapter 15]{MilStu05} for details):
\begin{itemize}
  \item $w(i)=i$ for $i\in\{1,\ldots,p\}$;
  \item $w(i)=i+1$ for $i\in\{p+1,\ldots,n\}$;
  \item $w(n+1)=p+1$.
\end{itemize}
In that case, the Grothendieck polynomial $\mathcal G_w$ associated to $w$ is a polynomial in $\Z[t_1,\ldots, t_{p+1}]$. Its evaluation at 
$t_1=t^{d_0-1},\ldots, t_{p+1}=t^{d_p-1}$ yields the desired numerator of the weighted Hilbert series of $\Did$. This representation provides more combinatorial insights, and it would be interesting to investigate if the formula for the degree of regularity can also be obtained from the evaluation of the Grothendieck polynomials. Moreover, this approach extends to matrices of corank greater than $1$ while the Eagon-Northcott complex is restricted to the case of maximal minors. We refer the reader to \cite{KnuMil05} for more details.

\medskip

{\bf Example.} Set $n=3$, $p=1$, $d_0=3$, $d_1=2$. This corresponds to the problem of minimizing a cubic function in three variables on a quadric surface. The ideal $\Did\in\Q[U_{0,1},U_{0,2},U_{0,3},U_{1,1},U_{1,2},U_{1,3}]
$ is generated by the $2$-minors of the matrix $(U_{i,j})$. The grading is given by
$\deg(U_{0,i})=d_0-1=2$ and $\deg(U_{1,i})=d_1-1=1$ for $i\in
\{1,2,3\}$.
In that case, Proposition \ref{prop:HSDmixed} yields
$$\begin{array}{rcl}\wHS_{\Q[U]/\Did}(t)&=&\dfrac{1-\left[\underset{0\leq k\leq 1}{\sum}\left[(-1)^k\underset{i_0+i_1=k}{\sum}\binom{3}{2+k}t^{2 i_0+i_1+3}\right]\right]}{(1-t^{2})^3(1-t)^3}\\
&=&\dfrac{t^5+t^4-3 t^3+1}{(1-t^{2})^3(1-t)^3}.
\end{array}
$$

The permutation $w\in S_4$ associated to the ideal $\Did$ is
given by $w(1)=1, w(2)=3, w(3)=4, w(4)=2$.
The corresponding Grothendieck polynomial is
$$\mathcal G_w(t_1,t_2)=t_1^2 t_2 + t_1 t_2^2 -3 t_1 t_2 +1.$$
Evaluating $\mathcal G_w$ at $(t^2,t)$ recovers the $K$-polynomial of the determinantal ideal:
$$\mathcal G_w(t^2,t)=t^5 + t^4 -3 t^3 +1.$$

\subsection{Complexity analysis}
\label{sec:critpointsmixedcompl}

We bound in this section the complexity of the following general solving strategy: first, one computes a Gr\"obner basis of $\Icrit(q,\F)$ with respect to the graded reverse lexicographical ordering (\emph{grevlex} for short) with the $F_4$/$F_5$ algorithm. Then, the FGLM algorithm is used to convert it into a lexicographical Gr\"obner basis. Once a lexicographical Gr\"obner basis is known, a rational parametrization of the critical points can be computed, for instance with the RUR algorithm \cite{Rou99}. Since the most costly steps of the solving process are the Gr\"obner bases computations, we focus in this paper on their complexities.

First, we need to estimate the complexity in terms of the degree and of the witness degree. Since Gr\"obner bases computations can be reduced to the computation of row echelon forms of Macaulay matrices, we have the following estimate:

\begin{theorem}\label{thm:complF4}
  A grevlex Gr\"obner basis of $\Icrit(\frakqa,\Fgena)$ can be computed within
$$O\left(\left(p+\binom{n}{p+1}\right) \binom{n+\dwit}{n}^\omega\right)$$ arithmetic operations in $\Q$, where $\dwit\leq\dreg(\Icrit(\frakq,\mathbf F^\infty))$
and $\omega$ is a feasible exponent for the matrix multiplication ($\omega<2.373$ with Williams' algorithm \cite{Vas11}).
\end{theorem}

\begin{proof}
Postponed to Section \ref{sec:proofcompl}.
\end{proof}

Finally we can obtain a general formula for the complexity of computing a lexicographical Gr\"obner basis of $\Icrit(\F,q)$ in terms of the generic values of the degree and of the witness degree:

\begin{corollary}\label{coro:complGB}
  Let $p,n\in\N$ with $p<n$, $(d_0,\ldots,d_p)\in\N^p$, and $(q,\F)\in\Q[X]^{p+1}$ be a generic system of respective degrees at most $(d_0,\ldots, d_p)$. Then the complexity of computing a lexicographical Gr\"obner basis of $\Icrit(q,\F)$ is bounded above by
$$O\left(\left(p+ \binom{n}{p+1}\right)\binom{n+\dwit}{n}^\omega+n\cdot\delta^3\right),$$
where $$
\begin{array}{rcl}
\dwit&\leq &\displaystyle(n-p-1) \max\{d_i-1\} -n-p+d_0 + 2\sum_{1\leq i\leq p} d_i\\
\delta&=&\displaystyle\left(\prod_{1\leq i\leq p} d_i\right)
\sum_{i_0+\dots+i_p=n-p} (d_0-1)^{i_0} \dots (d_p-1)^{i_p}.
\end{array}
$$
\end{corollary}

\begin{proof}
  This is a direct consequence of Theorem \ref{thm:complF4}, Corollary \ref{coro:dmaxdreg}, Corollary \ref{coro:dreg}, of the complexity of the FGLM algorithm $O\left(n\cdot\delta^3\right)$ \cite[Proposition 4.1]{FauGiaLazMor93} and from the explicit formula for the algebraic degree of polynomial optimization \cite[Theorem 2.2]{NieRan09}.
\end{proof}

In what follows, $A$ (resp. $G$) is the arithmetic (resp. geometric) average of the
multiset
$$\begin{array}{rcl}&&\{d_1,\ldots, d_p,\underbrace{\underset{0\leq i\leq p}\max\{d_i-1\},\dots,\underset{0\leq i\leq p}\max\{d_i-1\}}_{n-p}\},\\
  A&=&\displaystyle\frac{1}{n}\left((n-p)\max_{0\leq i\leq p}\{d_i-1\}+\sum_{1\leq i\leq p}d_i\right)\\
  G&=&\displaystyle\left(\max_{0\leq i\leq p}\{d_i-1\}^{n-p}\prod_{1\leq i\leq
    p}d_i\right)^{1/n}.\end{array}$$

Also, we let $\delta$ denote the generic algebraic degree of polynomial optimization \cite{NieRan09}:
$$\delta=\DEG(\Icrit(q,\F))=\left(\prod_{1\leq i\leq p} d_i\right) \sum_{i_0+\dots+i_p=n-p} (d_0-1)^{i_0} \dots (d_p-1)^{i_p}.$$

The next statement is the main result of this paper and bounds the complexity in terms of $\delta$:

\begin{theorem}\label{thm:complpfixed}
Let $(q,f_1,\ldots, f_p)\in \Q[X]$ be a generic system of polynomials of respective degrees at most $(d_0,\ldots, d_p)$ with $1\leq p<n$, $d_0\geq 1$, $d_1,\ldots, d_p\geq 2$. The complexity of computing a lexicographical Gr\"obner basis of $\Icrit(q,\F)$ is bounded above by
$\delta^{O\left(\log(A)/\log(G)\right)}.$
\end{theorem}

\begin{proof}
  First note that the number of sequences $(i_0,\ldots, i_p)$ such that $i_0=0$ and $\sum_{j=0}^p i_j=n-p$ is $\binom{n-1}{p-1}$. Consequently, the inequality $\delta\geq 2^p\binom{n-1}{p-1}$ holds, and hence the algorithm FGLM is polynomial in $\delta$ since its
  complexity is $O(n\cdot\delta^3)$ \cite[Thm.~5.1]{FauGiaLazMor93} and $n\leq 2\binom{n-1}{p-1}\leq\delta$. It is thus sufficient to prove that 
  a grevlex Gr\"obner basis of $\Icrit(q,\F)$ can be computed within
  $\delta^{O\left(\log(A)/\log(G)\right)}$ arithmetic operations. 

Next, the same inequality $\delta\geq 2^p\binom{n-1}{p-1}$ yields
$$p+\binom{n}{p+1}=\delta^{O(1)}.\label{eq:ComplB}$$
Since $\dwit+n<2A n$, we obtain
\begin{empheq}{align}\displaystyle O\left(\left(p+\binom{n}{p+1}\right)\binom{n+\dwit}{n}^\omega\right)\leq&\displaystyle \delta^{O(1)}\cdot O\left(
\binom{2 An}{n}^\omega\right)\notag\\
  \leq& \displaystyle \delta^{O(1)}\cdot O\left(\frac{(2An)^{\omega n}}{(n!)^\omega}\right)\notag\\
  \leq&\delta^{O(1)}\cdot A^{O\left(n\right)}\notag.\label{eq:ComplA}\end{empheq}
Finally, using the fact that $\delta\geq G^n$, we obtain
$$\frac{\log\left({\left(p+\binom{n}{p+1}\right)\binom{n+\dwit}{n}^\omega}\right)}{\log \delta}=O\left(\frac{n\log A}{\log \delta}\right)= O\left(\frac{n\log A}{n\log G}\right)=O\left(\frac{\log A}{\log G}\right).$$
\end{proof}

The next statement shows that this complexity meets the best known complexity bound $D^{O(n)}$. Note that the codimension $p$ does not appear in the following complexity bound: this comes from the fact that $p\leq n$ and hence the dependency in $p$ is hidden in the $O(n)$.
\begin{corollary}\label{coro:complDbounded}
Set $D=\max_{0\leq i\leq p}\{d_i\}$. If $D\geq 2$ and with the same notations and the same genericity assumptions as in Theorem \ref{thm:complpfixed}, the complexity of computing a lexicographical Gr\"obner basis of $\Icrit(q,\F)$ is bounded above by
$D^{O(n)}$.
\end{corollary}

\begin{proof}
By Corollary \ref{coro:complGB}, we have
$$\begin{array}{rcl}\dwit&\leq&\displaystyle(n-p-1) \max\{d_i-1\} -n-p+d_0 + 2\sum_{1\leq i\leq p} d_i\\
&\leq&\displaystyle(n-p-1) (D-1) -n-p+D + 2p D.\end{array}$$
Next, a proof exactly similar to that of Theorem \ref{thm:complpfixed} shows that the complexity is bounded above by
$$\delta^{O(1)}\cdot \left(\dfrac{n-p}{n}(D-1)+\dfrac pn D\right)^{O(n)}=\delta^{O(1)}\cdot D^{O(n)}.$$
The proof is concluded by noticing that $\displaystyle\delta\leq D^p(D-1)^{n-p}\binom np\leq (2 D)^n=D^{O(n)}$.
\end{proof}

In several applications, $p$ is small compared to $n$. Recall that the size of the Gr\"obner basis is polynomial in $\delta$ \cite[Coro.~2.1]{FauGiaLazMor93}. Although the following estimate is sometimes worse than the one derived in Corollary \ref{coro:complDbounded} (for instance when $p=n-1$), it shows that the complexity is polynomial in $\delta$ for subfamilies of problems where $p$ grows sufficiently slowly with~$n$:

\begin{corollary}\label{coro:cornp}
If $\max\{d_i\}\geq 3$, then $\displaystyle\frac{\log(A)}{\log(G)}= O\left(\frac{n}{n-p}\right)$, and hence the complexity bound in Theorem \ref{thm:complpfixed} can be specialized to 
$$\delta^{O\left(n/(n-p)\right)}.$$
Consequently, if $p<\alpha\, n$ for $0<\alpha<1$, the complexity of computing a lexicographical Gr\"obner basis of $\Icrit(q,\mathbf F)$ (where $q,\mathbf F$ is a generic system) is bounded above by $\delta^{O\left(\frac{1}{1-\alpha}\right)}$.
\end{corollary}

\begin{proof}
The first statement is a direct consequence of Theorem \ref{thm:complpfixed} and of the following inequalities
$$
\begin{array}{rcl}
  \log(A)&=&\displaystyle\log\left(\frac 1n\left[\sum_{i=1}^p d_i + (n-p)\max_{0\leq i\leq p}\{d_i-1\}\right]\right)\\
  &\leq&\displaystyle\log\left(\max_{0\leq i\leq p}\{d_i\}\right).\\
\log(G)&=&\displaystyle\frac 1n\left(\sum_{i=1}^p \log(d_i) + (n-p)\log(\max_{0\leq i\leq p}\{d_i-1\})\right)\\
&\geq& \displaystyle\frac{n-p}{n}\log(\max_{0\leq i\leq p}\{d_i-1\}).
\end{array}
$$
Since $\max\{d_i\}\geq 3$, we obtain $\log(A)/\log(G)\leq \frac{\log_2(3)\,n}{n-p}$.
The second statement is a direct consequence of the first statement: if $p<\alpha n$, then $n/(n-p)<1/(1-\alpha)$.
\end{proof}

The next corollary shows that in the context of quadratic programming, the complexity is polynomial in $n$. Such a bound was already obtained by a different approach in \cite{FauSafSpa12}.

\begin{corollary}[quadratic programming]
If $d_0=\dots=d_p=2$ and $(q,\F)$ is a generic quadratic system, then the complexity of computing a lexicographical Gr\"obner basis of $\Icrit(q,\F)$ is bounded by $n^{O(p)}$.
\end{corollary}

\begin{proof}
If $d_0=\dots=d_p=2$, then $\dwit\leq 2p+1$. The complexity bound in Corollary \ref{coro:complGB} for computing a lexicographical Gr\"obner basis gives
$$O\left(\left(p+ \binom{n}{p+1}\right)\binom{n+2p+1}{n}^\omega+n\cdot\delta^3\right).$$
Since we have
$$
\begin{array}{rcl}
  \displaystyle p+ \binom{n}{p+1}&=&\displaystyle n^{O(p)},\\
  \displaystyle\binom{n+2p+1}{n}&=&\displaystyle n^{O(p)},\\
  \delta&=&2^p\binom np=\displaystyle n^{O(p)},
\end{array}$$
we obtain that the total complexity is bounded by $\displaystyle n^{O(p)}$.
\end{proof}

\section{Experimental results}\label{sec:expe}

The goal of this section is to provide experimental evidence that the
asymptotic complexity results proved in Section
\ref{sec:critpointsmixedcompl} holds in practice for tractable sets of
parameters. We use the software {\tt FGb 1.58}\footnote{{\tt Maple} package available at
  http://www-polsys.lip6.fr/\textasciitilde jcf/Software/FGb/} to
compute the grevlex Gr\"obner basis.

\medskip

{\bf Workstation and experimental setting.} All computations have been
performed on an {Intel Core i5-3570 3.4GHz} processor. Since we wish to count the number of
arithmetic operations, all computations are done over the finite field
$\GF(65521)$ so that there is no effect of the growth of the
coefficients on the timings. Instances are generated as follows: for
$p,n\in\N$, $(d_0,\ldots, d_p)\in\N^{p+1}$ we pick inhomogeneous
polynomials $q,f_1,\ldots, f_p\in\GF(65521)[X_1,\ldots, X_n]$ of
respective degree $d_0,\ldots, d_p$ uniformly at random. Then, we
compute a grevlex Gr\"obner basis of the ideal $\Icrit(q,(f_1,\ldots,
f_p))$ with {\tt FGb}. For all tests, $\max_{0\leq i\leq p}\{d_i\}\geq 3$.

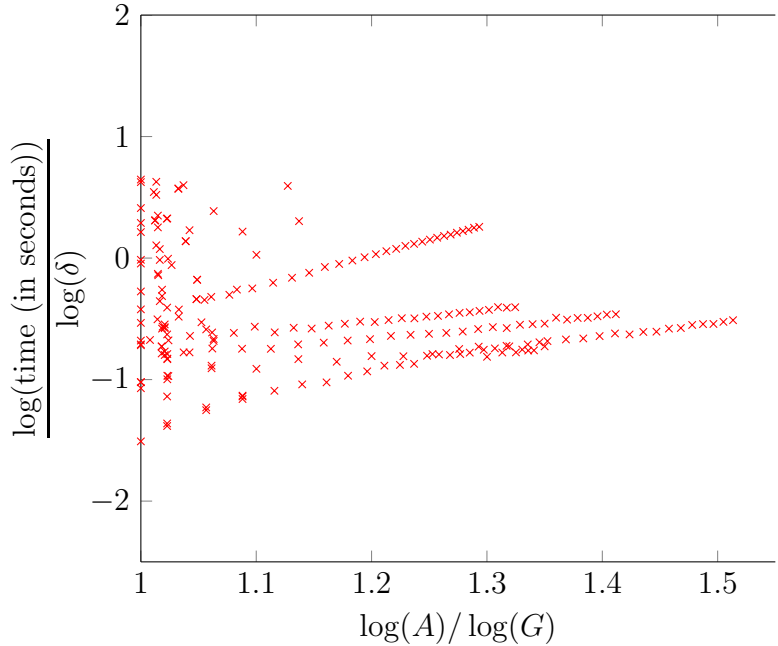
\begin{figure}
\centering
\begin{tikzpicture}[scale=1]
\pgfplotsset{every axis legend/.append style={at={(0.1,0)},anchor=south west}}
\begin{axis}[scale only axis,xmin=1,ymin=-2.5, ymax=2, xmax=1.55,axis y line*=left,xlabel={\normalsize $\log(A)/\log(G)$},
legend columns=3, ylabel={\normalsize $\dfrac{\log(\text{time (in seconds)})}{\log(\delta)}$},
]
\addplot[red, mark=x,only marks] table{logthm3.6};
\end{axis}
\end{tikzpicture}\caption{Experimental verification of the complexity bound in Theorem \ref{thm:complpfixed}\label{fig:mainthm}}
\end{figure}

\smallskip 

{\bf Experimental verification of Theorem \ref{thm:complpfixed}.} Figure \ref{fig:mainthm} shows the behavior of the logarithm of the complexity of the grevlex Gr\"obner basis computation with {\tt FGb} in terms of $\log(A)/\log(G)$. Theorem \ref{thm:complpfixed} states that $\log(timing)/\log(\delta)$ should be linear in $\log(A)/\log(G)$ which seems to be validated by experiments.
\begin{figure}
\centering
\begin{tikzpicture}[scale=1]
\pgfplotsset{every axis legend/.append style={at={(0.1,0)},anchor=south west}}
\begin{axis}[scale only axis,xmin=1,ymin=-5, ymax=6, xmax=11,axis y line*=left,xlabel={\normalsize $n$},
legend columns=3, ylabel={\normalsize $\dfrac{\log(\text{time (in seconds)})}{\log(D)}$},
]
\addplot[red, mark=x, thick,only marks] table{logcoro3.7};
\end{axis}
\end{tikzpicture}\caption{Experimental verification of the complexity bound in Corollary \ref{coro:complDbounded}\label{fig:verif}}
\end{figure}
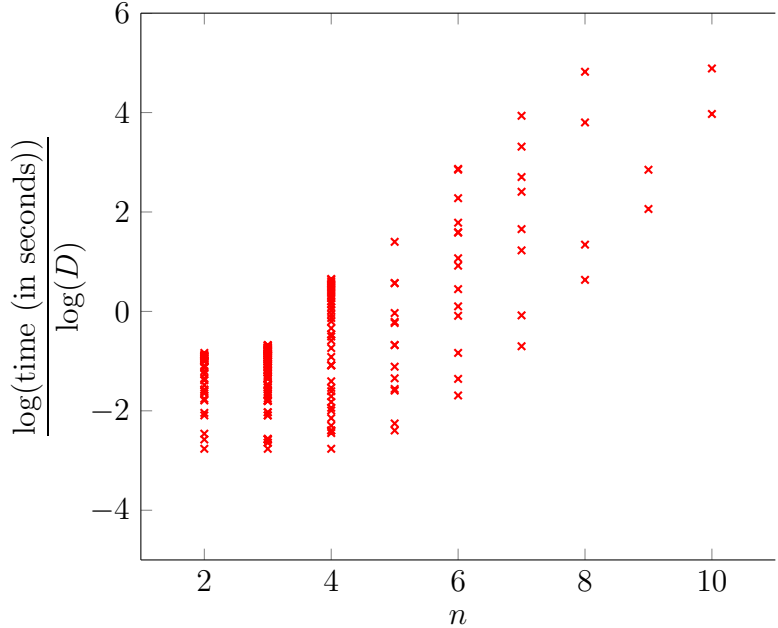

\smallskip 

{\bf Experimental verification of Corollary \ref{coro:complDbounded}.} Figure \ref{fig:verif} shows the behavior of the logarithm of the complexity in terms of the maximum of the degrees of the input system. The figure seems to indicate that $\log(timing)/\log(D)$ is linear in $n$, where $D=\max\{d_0,\ldots, d_p\}$. This provides experimental evidence of Corollary \ref{coro:complDbounded}, namely that the complexity is bounded above by $D^{O(n)}$.

\section{Proof of Theorem \ref{thm:complF4}}
\label{sec:proofcompl}
One method to bound the witness degree and the complexity of computing
Gr\"obner bases when the input polynomials are inhomogeneous is to
bound the degree of the polynomials in a grevlex Gr\"obner basis of
the homogenized system (by introducing a homogenization variable). In the sequel, we use the following
notations:

\begin{notation}
\begin{itemize}
\item The ring $\Q[U_{0,1},\ldots,U_{p,n},X_1,\ldots,X_n,H]$ with grading $\deg(U_{i,j})=d_i-1$, $\deg(X_i)=\deg(H)=1$ is denoted by $\Shom$;
\item for any polynomial $\frakfa\in\K[X]$, we let $\frakfhom\in\Shom$ denote its homogenization:
$$\frakfhom=H^{\deg(\frakfa)}\cdot\frakfa\left(\frac{X_1}H,\ldots,\frac{X_n}H\right)\in\Shom;$$
\item For $i\in\{0,\ldots,p\}, j\in\{1,\ldots,n\}$, let $\frakghom_{i n+j}\in\Shom $ denote the polynomial:
$$\frakghom_{i n+j}=
\begin{cases}
  U_{i,j}-\frac{\partial \frakfhom_i}{\partial X_j}\text{ if $i\in\{1,\ldots,p\}$},\\
  U_{i,j}-\frac{\partial \frakqhom}{\partial X_j}\text{ if $i=0$};
\end{cases}
$$
\item for $i\in\{1,\ldots,p\}$, we set $\frakghom_{(p+1)n+i}=\frakfhom_i$.
\end{itemize}
\end{notation}

The two following statements show that the algebraic structure of the ideal generated by the homogenized critical system is the same as the structure of the ideal generated by the homogeneous components of highest degree.

\begin{corollary}\label{coro:Hilberthomm} The following equality holds for a generic system $(q,f_1,\ldots, f_p)$ of degrees at most $(d_0,\ldots, d_p)$:
 $$\HS_{\K[X,H]/\Icrit(\frakqhom,\mathbf F^h)}(t)=\HS_{\K[X,H]/\Icrit(\frakq,\Fgenh)}(t).$$
\end{corollary}

\begin{proof}
First, note that the rings $\K[X]/\Icrit(\frakq,\Fgenh)$ and $\K[X,H]/(\Icrit(\frakqhom,\mathbf F^h)+\langle H \rangle)$ are isomorphic. 
Consequently, they also share the same Krull dimension, which is $0$ by Lemma \ref{lem:dim0}. Therefore, $\K[X,H]/\Icrit(\frakqhom,\mathbf F^h)$ has dimension at most $1$ since quotienting by $H$ can only decrease the dimension by one. Since $\K[X,H]/\Icrit(\frakqhom,\mathbf F^h)$ is isomorphic to $\Q[U,X,H]/(\Did+\langle g_1^h,\ldots,g_{p+n(p+1)}^h\rangle)$, this latter ring has also Krull dimension at most $1$. Next, note that $\Q[U,X,H]$ is a polynomial ring in $n+n(p+1)+1$ variables. Consequently, by Macaulay's Unmixedness Theorem and similarly to the proof of Lemma \ref{lem:nondivgl}, $g_1^h,\ldots,g_{p+n(p+1)}^h$ is a regular sequence in $\Q[U,X,H]/\Did$. Therefore, the Hilbert series of $\K[X,H]/(\Icrit(\frakqhom,\mathbf F^h)$ is
$$\wHS_{\K[X,H]/(\Icrit(\frakqhom,\mathbf F^h)}(t)=\wHS_{\Q[U,H]/\Did}(t)\cdot\frac{(1-t^{d_{0}-1})^n\prod_{1\leq i\leq p} (1-t^{d_i})(1-t^{d_i-1})^n}{(1-t)^n}.$$
Finally, by the same proof as Lemma \ref{coro:HSImixed}, we obtain for $\K[X,H]/\Icrit(\frakq,\Fgenh)$
$$\wHS_{\K[X,H]/\Icrit(\frakq,\Fgenh)}(t)=\wHS_{\Q[U,H]/\Did}(t)\cdot\frac{(1-t^{d_{0}-1})^n\prod_{1\leq i\leq p} (1-t^{d_i})(1-t^{d_i-1})^n}{(1-t)^n},$$
which concludes the proof.
\end{proof}

The next statement relates the degree of regularity of $\Icrit(\frakq,\Fgenh)$ with the maximal degree in the reduced grevlex Gr\"obner basis of the homogenized system $\Icrit(\frakqhom,\mathbf F^h)$.

\begin{corollary}\label{coro:dmaxdreg}
  For any homogeneous ideal $I\subset S$, let $\dmax_S(I)$ denote the maximal degree of a polynomial in the reduced grevlex Gr\"obner basis of $I$. If $q,f_1,\ldots, f_p$ are generic polynomials, then 
$$\dmax_{\Shom}(\Icrit(\frakqhom,\mathbf F^h))=\dmax_S(\Icrit(\frakq,\Fgenh))=\dreg(\Icrit(\frakq,\Fgenh)).$$
\end{corollary}

\begin{proof}
  For $f$ in $S$ or $\Shom$, let $\LM(f)$ denote its leading monomial
  with respect to the grevlex ordering with $h\prec X_n\prec\dots\prec
  X_1$.  For any polynomial $f^h\in \Shom$ not divisible by $h$,
  $\LM(f^h(X_1,\ldots, X_n,0))=\LM(f^h(X_1,\ldots,
  X_n,H))$. Consequently
  $\LM(\Icrit(\frakq,\Fgenh))\subset\LM(\Icrit(\frakqhom,\mathbf
  F^h))$. By Corollary \ref{coro:Hilberthomm} and since for any
  homogeneous ideal $I$, $\wHS_{S/I}=\wHS_{S/\LM(I)}$, we obtain
  $\LM(\Icrit(\frakqhom,\mathbf F^h))=\LM(\Icrit(\frakq,\Fgenh))$. The
  degrees of the polynomials in the reduced Gr\"obner basis of a
  homogeneous ideal $I$ equal the degrees of a minimal set of
  generators of $\LM(I)$. Consequently,
  $\dmax_{\Shom}(\Icrit(\frakqhom,\mathbf
  F^h))=\dmax_S(\Icrit(\frakq,\Fgenh))$. The second equality
  $\dmax_S(\Icrit(\frakq,\Fgenh))=\dreg(\Icrit(\frakq,\Fgenh))$ is a
  consequence of the definition of the degree of regularity (see
  Section~\ref{sec:notations}).
\end{proof}

We can now conclude the proof of the complexity of the grevlex Gr\"obner basis computation:
\begin{proof}[of Theorem \ref{thm:complF4}]
Recall that the witness degree $\dwit$ is defined as the smallest
integer $d$ such that the $\Q$-vector space
$$T_d=\left\{\sum f_i \alpha_i +\sum m_j \beta_j \mid \alpha_i,\beta_j\in
  \Q[X_1,\ldots, X_n], \deg(f_i \alpha_i)\leq d,\deg(m_j \beta_j)\leq
  d\right\}$$ contains the reduced grevlex Gr\"obner basis of
$\Icrit(q,\F)$.  It is known that dehomogenizing a grevlex Gr\"obner
basis of $\Icrit(q^h,\F^h)$ yields a Gr\"obner basis of
$\Icrit(q,\F)$.  A consequence of this fact is that
$\dwit\leq\dmax_{\Shom}(\Icrit(\frakqhom,\mathbf
F^h))$ (see \emph{e.g.}
\cite{Laz83,BarFauSal04}\cite[Proposition 16]{FauSafSpa11a}). By Corollary \ref{coro:dmaxdreg}, $\dmax_{\Shom}(\Icrit(\frakqhom,\mathbf
F^h))=\dreg(\Icrit(\frakq,\Fgenh))$

Now the goal is to compute a triangular basis of the $\Q$-vector space
$T_{\dreg(\Icrit(q^\infty,\F^\infty))}$.  This vector space is equal to the
row span of the \emph{Macaulay matrix} in degree
$\dreg(\Icrit(q^\infty,\F^\infty))$, which is constructed as follows.
The rows of the Macaulay matrix are indiced by all products $\mu\cdot
k_\ell$, where $k_\ell$ is a polynomial of the system generating
$\Icrit(q,\F)$ (\emph{i.e.} either an input polynomial $f_i$ or a maximal
minor of $\jac(q,\F)$) and $\mu$ ranges through all monomials
of degree at most $\dreg(\Icrit(q^\infty,\F^\infty))-\deg(k_\ell)$.
The columns of this matrix are indiced by all the monomials of degree
at most $\dreg(\Icrit(q^\infty,\F^\infty))$. The entries of a row of
the matrix are the coefficients of the corresponding polynomial
$\mu\cdot k_\ell$. The number of rows (resp. columns) of
the Macaulay matrix is bounded above by
$\left(n+\binom{n}{p+1}\right)\binom{n+\dreg(\Icrit(q^\infty,\F^\infty))}{n}$
(resp. $\binom{n+\dreg(\Icrit(q^\infty,\F^\infty))}{n}$).  Since the
row echelon form of a $A\times B$ matrix can be computed within $O(A B
\min(A,B)^{\omega-2})$ operations \cite[Prop.~2.11]{Sto00}, (where
$\omega$ is a feasible exponent for matrix multiplication), this
computation can be performed within
$$\displaystyle O\left(\left(n+\binom{n}{p+1}\right)\binom{n+\dreg(\Icrit(q^\infty,\F^\infty))}{n}^\omega\right)$$ arithmetic operations in $\Q$.
The polynomials corresponding to the rows of the reduced Macaulay
matrix yield a grevlex Gr\"obner basis of $\Icrit(q,\F)$.
\end{proof}

\bibliographystyle{abbrv}
\bibliography{biblioPtscritmixed}

\end{document}